\documentclass[11pt,a4paper]{article}
\usepackage[utf8]{inputenc}
\usepackage[T1]{fontenc}
\usepackage{amsmath, amssymb, amsthm, mathtools}
\usepackage{geometry}
\usepackage[round]{natbib}
\usepackage{hyperref}
\usepackage{booktabs}
\newtheorem{theorem}{Theorem}
\newtheorem{lemma}[theorem]{Lemma}
\newtheorem{proposition}[theorem]{Proposition}
\newtheorem{corollary}[theorem]{Corollary}
\newtheorem{conjecture}[theorem]{Conjecture}
\theoremstyle{definition}

\newtheorem{assumption}[theorem]{Assumption}
\newtheorem{example}[theorem]{Example}
\theoremstyle{remark}
\newtheorem{remark}{Remark}

\title{Structural properties of the implicit function defined\\
by an integral self-consistency equation}

\author{Ivan Viakhirev\\
\small ITMO University, St.\ Petersburg, Russian Federation\\
\small \texttt{i.viakhirev@mail.ru}}
\date{}

\begin{document}
\maketitle

\begin{abstract}
We study the integral equation
\[
\int_0^m \frac{\eta\,\rho(\eta)}{C-\eta}\,d\eta = 1, \qquad C>m,
\]
where $\rho$ is a $C^1$ probability density on $[0,M]$ vanishing
polynomially at $\eta=M$. Setting
$\mathcal{I}^+(m):=\lim_{C\downarrow m}\int_0^m \eta\rho(\eta)/(C-\eta)\,d\eta$
and $\Omega:=\{m\in(0,M):\mathcal{I}^+(m)>1\}$, the displayed
equation determines $C$ implicitly as a function of $m$ on $\Omega$,
and our object of study is the dimensionless ratio $\beta(m):=C(m)/m$.

Write $h(\eta):=\eta\rho(\eta)$. Our main result
(Theorem~\ref{thm:structure}) establishes openness of $\Omega$,
$C^1$-smoothness of $\beta$, a sign formula identifying $\beta'(m)$
with a positively-weighted integral of $d h/d\ln\eta$, transfer of
monotonicity from $h$ to $\beta$, and existence of an interior
critical point of $\beta$ when $h$ is unimodal and two technical
hypotheses hold.

Numerically, $\beta$ has a single critical point in seven log-concave
test densities (mostly Beta-type), in support of a separate
uniqueness conjecture. A bimodal density that violates both
unimodality and log-concavity exhibits three critical points; this
shows that dropping the two hypotheses jointly admits multiple
critical points, but does not separate their roles. Applications to
preferential-attachment networks are mentioned only as motivation.
\end{abstract}

\section{Introduction}\label{sec:intro}

\subsection{The equation and the object of study}

Let $M\in(0,\infty)$ and let $\rho$ be a probability density on $[0,M]$. For
$m\in(0,M]$ and $C>m$, define
\begin{equation}\label{eq:F-def}
F(C,m) \;:=\; \int_0^m \frac{\eta\,\rho(\eta)}{C-\eta}\,d\eta - 1.
\end{equation}
We study the equation $F(C,m)=0$.

When this equation has a unique solution $C(m)>m$, the ratio
\begin{equation}\label{eq:beta-def}
\beta(m) \;:=\; \frac{C(m)}{m}
\end{equation}
is the central object of this paper. The set of $m$ for which a solution exists
is characterized by the limit
\begin{equation}\label{eq:I-def}
\mathcal{I}^+(m) \;:=\; \lim_{C\downarrow m}\,(F(C,m)+1) \;=\; \lim_{C\downarrow m}\int_0^m\frac{\eta\rho(\eta)}{C-\eta}d\eta \;\in\; (0,\infty],
\end{equation}
which exists by monotone convergence (the integrand increases as $C\downarrow m$).
When $\mathcal{I}^+(m)<\infty$, the limit equals the improper integral
$\int_0^m\eta\rho(\eta)/(m-\eta)d\eta$.

\subsection{Origin and motivation}

Equations of the form~\eqref{eq:F-def} arise in the mean-field
analysis of preferential-attachment dynamics with heterogeneous
fitness; see Section~\ref{sec:discussion} for references and for the
limits of the analogy. In these applications it is the ratio
$\beta=C/m$, rather than $C(m)$, that carries an intrinsic meaning:
it is dimensionless, it is bounded below by $1$ throughout the
domain, and it controls the relevant degree-distribution exponent.
Analytically, the substitution $\eta=mt$ in~\eqref{eq:F-def} turns
the equation into one for $\beta$ alone on the fixed interval
$[0,1]$ (Lemma~\ref{lem:reparam}), which is what makes the
derivative formula (Theorem~\ref{thm:structure}(c)) clean.

\subsection{Main results}

Write $h(\eta):=\eta\rho(\eta)$. Under
hypotheses~(D0)--(D3) on $\rho$ (Section~\ref{sec:hyp}),
Theorem~\ref{thm:structure} establishes the following five
properties of $\beta$.
\begin{itemize}
\item[(a)] \emph{Domain.} $\beta(m)$ is defined exactly on the open
set $\Omega:=\{m\in(0,M):\mathcal{I}^+(m)>1\}\subset(0,M)$.
\item[(b)] \emph{Smoothness.} $\beta\in C^1(\Omega)$ and $\beta>1$ on $\Omega$.
\item[(c)] \emph{Sign formula.} For $m\in\Omega$,
\begin{equation}\label{eq:sign-intro}
\mathrm{sign}\bigl(\beta'(m)\bigr)
\;=\; \mathrm{sign}\!\left(\int_0^1 \frac{1}{\beta(m)-t}\cdot\frac{dh}{d\ln\eta}\Big|_{\eta=mt}\,dt\right).
\end{equation}
\item[(d)] \emph{Monotonicity transfer from $h$ to $\beta$.} On any
sub-interval of $\Omega$ where $h$ is monotone and non-constant,
$\beta$ is strictly monotone in the same direction.
\item[(e)] \emph{Critical point under unimodality.} Suppose $h$ is
unimodal with mode $\eta^*_h\in(0,M)$, and let $\mathcal{J}(m)$
denote the integral on the right-hand side
of~\eqref{eq:sign-intro}. If there is $m^\sharp\in(\eta^*_h,M)$ with
$[\eta^*_h, m^\sharp]\subset\Omega$ and $\mathcal{J}(m^\sharp)<0$,
then $\beta$ has a critical point $m^\star\in(\eta^*_h, m^\sharp)$.
\end{itemize}

The content of~\eqref{eq:sign-intro} is that $\beta'(m)$ has the
sign of a positive-weight average of the logarithmic derivative
$d h/d\ln\eta$ over $(0, m]$: $\beta$ inherits monotonicity from
the elasticity of $h$, and a critical point of $\beta$ is the
location at which positive and negative contributions of
$d h/d\ln\eta$, weighted by the implicit-function kernel
$1/(\beta(m)-t)$, balance. Part (d) is the case where the
elasticity has a definite sign throughout; part (e) is the simplest
case where it changes sign once. Part~(e) requires two technical
hypotheses, conventionally labelled
\begin{itemize}
\item[(H1)] $\mathcal{J}(m^\sharp)<0$ for some $m^\sharp\in(\eta^*_h,M)$,
\item[(H2)] $[\eta^*_h, m^\sharp]\subset\Omega$,
\end{itemize}
which we will refer to by these names. (H1) is the substantive
sign-change condition; (H2) is automatic whenever $\rho>0$ on
$(0,M)$ (in particular for log-concave $h$). We verify (H1)
numerically for all seven log-concave test densities of
Section~\ref{sec:numerics}.

\paragraph{A uniqueness conjecture.}
If, in addition to the hypotheses of~(e), $h$ is log-concave on
$(0,M)$, we conjecture that $m^\star$ is the unique critical point
of $\beta$ on $\Omega$ (Conjecture~\ref{conj:uniqueness}). A
bimodal example showing what happens when both hypotheses fail is
discussed in Section~\ref{sec:numerics}.

\section{Setup and preliminaries}\label{sec:setup}

\subsection{Hypotheses on $\rho$}\label{sec:hyp}

We use the following hypotheses throughout. Let $M\in(0,\infty)$ and let $\rho$
be a function on $[0,M]$.
\begin{itemize}
\item[\textup{(D0)}] $\rho\ge 0$ and $\int_0^M\rho(\eta)\,d\eta = 1$.
\item[\textup{(D1)}] $\rho\in C^1([0,M])$, with one-sided derivatives at the
endpoints.
\item[\textup{(D2)}] There exist $\alpha>1$, $K<\infty$ and $\eta_0\in(0,M)$
such that $\rho(\eta)\le K(M-\eta)^{\alpha-1}$ for all $\eta\in[\eta_0,M]$.
That is, $\rho$ vanishes polynomially at $\eta=M$ at order at least
$\alpha-1$; in particular $\rho(M)=0$. This holds for Beta-type
densities $\rho(\eta)\propto\eta^a(M-\eta)^b$ with $b\ge 1$ and fails
for any density with $\rho(M)>0$ (log-normal, exponential, Gamma);
the question of extending the analysis past (D2) is taken up in
Section~\ref{sec:discussion}.
\item[\textup{(D3)}] $\rho$ is not identically zero on any subinterval of $(0,M)$.
Equivalently, $\rho>0$ on a positive-measure subset of every
subinterval of $(0,M)$.
\end{itemize}
Define $h(\eta):=\eta\rho(\eta)$. By (D1), $h\in C^1([0,M])$, with
derivative $h'(\eta)=\rho(\eta)+\eta\rho'(\eta)$; by (D2), $h(M)=0$.

\subsection{Preliminary lemmas}

Theorem~\ref{thm:structure} relies on five lemmas: a finiteness
criterion for the limit $\mathcal{I}^+$ (Lemma~\ref{lem:I-finite}),
$C^1$-regularity and strict $C$-monotonicity of $F$
(Lemma~\ref{lem:F-smooth}), boundary values of $F$ in $C$
(Lemma~\ref{lem:F-boundary}), the reparametrisation $\eta=mt$ that
makes $\beta=C/m$ the only unknown (Lemma~\ref{lem:reparam}), and
openness of $\Omega$ (Lemma~\ref{lem:Omega-open}). Of these, only
Lemma~\ref{lem:I-finite} uses the polynomial-decay hypothesis~(D2).

\begin{lemma}[Finiteness criterion for $\mathcal{I}^+$]\label{lem:I-finite}
Under \textup{(D0), (D1), (D2)} with exponent $\alpha>1$:
\begin{enumerate}
\item[\textup{(i)}] At the endpoint $m=M$: $\mathcal{I}^+(M)<\infty$,
driven by the polynomial vanishing of $\rho$ supplied by (D2).
\item[\textup{(ii)}] At interior $m\in(0,M)$: $\mathcal{I}^+(m)<\infty$
iff $\rho(m)=0$, and $\mathcal{I}^+(m)=+\infty$ otherwise. The
integrability when $\rho(m)=0$ comes from the local Lipschitz bound
on $\rho$ implied by (D1).
\end{enumerate}
\end{lemma}

\begin{proof}
\emph{(i):} By monotone convergence (the integrand
$\eta\rho(\eta)/(C-\eta)$ increases pointwise as $C\downarrow M$ to
$\eta\rho(\eta)/(M-\eta)$),
\begin{equation}\label{eq:I-plus-M}
\mathcal{I}^+(M) \;=\; \int_0^M\frac{\eta\rho(\eta)}{M-\eta}\,d\eta,
\end{equation}
which is finite by the following splitting. On $[0,\eta_0]$ the
integrand $\eta\rho(\eta)/(M-\eta)$ is continuous on a compact set
with denominator $\ge M-\eta_0>0$, hence bounded. On $[\eta_0,M]$,
hypothesis (D2) gives
\begin{equation}\label{eq:D2-bound}
\frac{\eta\rho(\eta)}{M-\eta} \;\le\; \frac{MK(M-\eta)^{\alpha-1}}{M-\eta}
\;=\; MK(M-\eta)^{\alpha-2},
\end{equation}
and $\int_{\eta_0}^M(M-\eta)^{\alpha-2}\,d\eta<\infty$ since $\alpha>1$.
Hence $\mathcal{I}^+(M)<\infty$.

\emph{(ii), $\Leftarrow$:} Assume $\rho(m)=0$ with $m\in(0,M)$. By
(D1), $\|\rho'\|_\infty<\infty$, and the mean value theorem at
$\rho(m)=0$ gives $\rho(\eta)\le\|\rho'\|_\infty(m-\eta)$ on $[m/2,m]$.
For $C>m$, $(m-\eta)/(C-\eta)\le 1$, so on $[m/2,m]$ the integrand
$\eta\rho(\eta)/(C-\eta)$ is dominated by $m\,\|\rho'\|_\infty$;
on $[0,m/2]$ it is dominated by $\|\rho\|_\infty$ (because
$C-\eta\ge m/2$). The resulting constant dominator is $L^1$, so
dominated convergence as $C\downarrow m$ gives
\begin{equation}\label{eq:I-plus-interior}
\mathcal{I}^+(m) \;=\; \int_0^m\frac{\eta\rho(\eta)}{m-\eta}\,d\eta \;<\;\infty.
\end{equation}

\emph{(ii), $\Rightarrow$:} Assume $\rho(m)>0$. By continuity,
$\rho(\eta)\ge\rho(m)/2$ on $[m-\delta,m]$ for some $\delta>0$. For $C$
slightly greater than $m$,
\begin{equation}\label{eq:log-blowup}
\int_{m-\delta}^m\frac{\eta\rho(\eta)}{C-\eta}\,d\eta
\;\ge\; \frac{(m-\delta)\rho(m)}{2}\,
\ln\frac{C-(m-\delta)}{C-m} \;\longrightarrow\; +\infty
\quad\text{as } C\downarrow m,
\end{equation}
so $\mathcal{I}^+(m)=+\infty$.
\end{proof}

\begin{remark}[Implication for $\Omega$]\label{rem:Omega-structure}
Lemma~\ref{lem:I-finite}(ii) gives $\mathcal{I}^+(m)=+\infty>1$ at
every $m$ with $\rho(m)>0$, so
\begin{equation}\label{eq:Omega-bracket}
\{m\in(0,M):\rho(m)>0\}\;\subset\;\Omega
\;\subset\;\{m\in(0,M):\rho(m)>0\}\cup\{m:\rho(m)=0\}.
\end{equation}
By (D3), the lower set is dense in $(0,M)$. Whether $M$ is a limit
point of $\Omega$ from below depends on whether $\mathcal{I}^+(M)>1$,
a nontrivial extra condition.
\end{remark}

\begin{lemma}[Smoothness and strict monotonicity in $C$]\label{lem:F-smooth}
Under \textup{(D0), (D1)}:
\begin{enumerate}
\item[\textup{(i)}] $F\in C^1\bigl(\{(C,m): 0<m<M, C>m\}\bigr)$;
\item[\textup{(ii)}] $\partial_C F(C,m) = -\int_0^m \eta\rho(\eta)/(C-\eta)^2\,d\eta \le 0$;
\item[\textup{(iii)}] $\partial_m F(C,m) = m\rho(m)/(C-m)$.
\end{enumerate}
Under additionally \textup{(D3)}: $\partial_C F(C,m) < 0$ strictly for all $(C,m)$ in the domain.
\end{lemma}

\begin{proof}
\emph{Part (i)--(iii) under \textup{(D0), (D1)} only.}
We prove $C^1$ on each compact set
$K_{m_0,m_1,\varepsilon} := \{(C,m): m_0\le m\le m_1,\ m+\varepsilon\le C\le 1/\varepsilon\}$
for fixed $0<m_0<m_1<M$, $\varepsilon\in(0,m_0/4)$; the open domain is exhausted
by such sets. Since $C^1$-regularity is a local property, this suffices to
conclude $F\in C^1$ on the full open domain.

\emph{Continuity of $F$.} The integrand $\eta\rho(\eta)/(C-\eta)$ is continuous
in $(C,m,\eta)$ on the relevant set, and bounded by
$M\|\rho\|_\infty/\varepsilon$ uniformly. Dominated convergence gives continuity.

\emph{$\partial_C F$.} The pointwise derivative
$-\eta\rho(\eta)/(C-\eta)^2$ is bounded by
$M\|\rho\|_\infty/\varepsilon^2$, hence integrable, so
differentiation under the integral sign \citep[Thm.~2.27]{folland-1999} applies and
yields the formula for $\partial_C F$. The integrand
$\eta\rho(\eta)/(C-\eta)^2$ is non-negative by (D0), so
$\partial_C F\le 0$.

\emph{$\partial_m F$.} The integrand of $F$ does not depend on $m$
except through the upper limit; the Leibniz rule yields
$\partial_m F = m\rho(m)/(C-m)$, continuous since $\rho\in C^1$ and
$C-m\ge\varepsilon$.

Both partial derivatives are continuous on the compact set
$K_{m_0,m_1,\varepsilon}$ by dominated convergence, and a standard
result on functions of two variables (e.g.\ \citealp[\S2.5]{folland-1999})
gives $F\in C^1$ there. This establishes (i)--(iii) under
(D0), (D1).

\emph{Strict negativity under \textup{(D3)}.} Under (D3), $\rho>0$
on an open subinterval of $(0,m)$, on which the non-negative
integrand $\eta\rho(\eta)/(C-\eta)^2$ is strictly positive. The
integral is therefore strictly positive, giving $\partial_C F<0$
strictly throughout the domain.
\end{proof}

\begin{lemma}[Boundary values]\label{lem:F-boundary}
Under \textup{(D0), (D1)}, for fixed $m\in(0,M]$:
\begin{enumerate}
\item[\textup{(i)}] $\lim_{C\downarrow m} F(C,m) = \mathcal{I}^+(m) - 1$ (with $+\infty$ allowed);
\item[\textup{(ii)}] $\lim_{C\to\infty} F(C,m) = -1$.
\end{enumerate}
\end{lemma}

\begin{proof}
\emph{(i):} By definition of $\mathcal{I}^+$.

\emph{(ii):} For $C\ge 2m$, $\eta\in[0,m]$: $C-\eta\ge m$, so
$\eta\rho(\eta)/(C-\eta)\le \rho(\eta)$, integrable. Pointwise the integrand
$\to 0$. Dominated convergence gives $F\to-1$.
\end{proof}

\begin{lemma}[Reparametrization]\label{lem:reparam}
For $m>0$ and $C>m$, the substitution $\eta=mt$ gives $F(C,m)+1 = m\Phi(C/m, m)$, where
\begin{equation}\label{eq:Phi-def}
\Phi(\beta,m) \;:=\; \int_0^1 \frac{t\rho(mt)}{\beta-t}\,dt, \qquad \beta>1.
\end{equation}
Equivalently, $F(C,m)=0$ iff $m\Phi(\beta,m)=1$ where $\beta=C/m$.
\end{lemma}

\begin{proof}
The substitution $\eta=mt$, $d\eta=m\,dt$ sends $\eta\in[0,m]$ to $t\in[0,1]$, and
\begin{equation}\label{eq:reparam-comp}
\int_0^m\frac{\eta\rho(\eta)}{C-\eta}\,d\eta
\;=\; \int_0^1\frac{(mt)\rho(mt)\cdot m}{C-mt}\,dt
\;=\; m\int_0^1\frac{t\rho(mt)}{(C/m)-t}\,dt
\;=\; m\,\Phi(\beta,m).
\end{equation}
\end{proof}

\begin{lemma}[Openness of $\Omega$]\label{lem:Omega-open}
Define $\Omega:=\{m\in(0,M):\mathcal{I}^+(m)>1\}$. Under \textup{(D0), (D1), (D2), (D3)},
$\Omega$ is open in $(0,M)$.
\end{lemma}

\begin{proof}
Fix $m_0\in\Omega$.

\emph{Case $\mathcal{I}^+(m_0)<\infty$:} By Lemma~\ref{lem:F-boundary}(i),
$\lim_{C\downarrow m_0}F(C,m_0) = \mathcal{I}^+(m_0) - 1 > 0$, so by strict
inequality and continuity in $C$ there exists $C_0 > m_0$ with $F(C_0,m_0)>0$.

\emph{Case $\mathcal{I}^+(m_0)=+\infty$:} By
Lemma~\ref{lem:F-boundary}(i) (with $+\infty$ allowed), there exists
$\delta_0>0$ with $F(m_0+\delta_0, m_0)>0$. Set $C_0:=m_0+\delta_0$;
by Lemma~\ref{lem:F-smooth}, $F$ is continuous in a neighbourhood of
$(C_0,m_0)$, so there exists $\varepsilon>0$ with $F>0$ on the open
ball $B((C_0,m_0),\varepsilon)$. Pick
$\varepsilon'<\min(\varepsilon,\delta_0/2)$ (the bound $\delta_0/2$
keeps the projection $U:=\{m:|m-m_0|<\varepsilon'\}$ strictly below
$C_0$, so $C_0>m$ on $U$ and the pair $(C_0,m)$ stays in the domain
of $F$). Then $F(C_0,m)>0$ for all $m\in U$.

In both cases we obtain an open neighborhood $U\ni m_0$ in $(0,M)$ with
$\sup U < C_0$ and $F(C_0,m)>0$ for all $m\in U$.

For $m\in U$: $F(\cdot,m)$ is strictly decreasing in $C$
(Lemma~\ref{lem:F-smooth}(ii) under (D3)) with $F(C_0,m)>0$ and $F\to-1$ as $C\to\infty$.
Hence $\sup_{C>m}F(C,m)\ge F(C_0,m)>0$, which by Lemma~\ref{lem:F-boundary}(i)
equals $\mathcal{I}^+(m)-1$. Therefore $\mathcal{I}^+(m)>1$ and $m\in\Omega$.
\end{proof}

\section{The structural theorem}\label{sec:main}

\begin{assumption}[Unimodality]\label{ass:U}
$h(\eta)=\eta\rho(\eta)$ has a mode $\eta^*_h\in(0,M)$, in the sense that $h$ is
strictly increasing on $(0,\eta^*_h)$ and strictly decreasing on $(\eta^*_h, M)$.
\end{assumption}

Define
\begin{equation}\label{eq:J-def}
\mathcal{J}(m) \;:=\; \int_0^1 \frac{1}{\beta(m)-t}\cdot\frac{dh}{d\ln\eta}\Big|_{\eta=mt}\,dt
\;=\; m\int_0^1 \frac{t\,h'(mt)}{\beta(m)-t}\,dt,
\end{equation}
the second equality by $dh/d\ln\eta=\eta h'(\eta)$ at $\eta=mt$, which gives
$(dh/d\ln\eta)|_{\eta=mt}=mt\,h'(mt)$ and pulls the factor $m$ outside the
integral. The integrand of the second form is continuous on $[0,1]$ and
bounded by $\|h'\|_\infty/(\beta(m)-1)$, so $\mathcal{J}(m)$ is well-defined
wherever $\beta(m)$ is.

\begin{theorem}[Properties of the implicit ratio $\beta$]\label{thm:structure}
Assume \textup{(D0), (D1), (D2), (D3)} with exponent $\alpha>1$, and set
\begin{equation}\label{eq:Omega-def}
\Omega \;:=\; \{m\in(0,M) : \mathcal{I}^+(m)>1\}.
\end{equation}
Then:
\begin{enumerate}
\item[\textup{(a)}] \emph{(Domain.)} $\Omega$ is open in $(0,M)$. For each
$m\in\Omega$, equation $F(C,m)=0$ has a unique solution $C(m)\in(m,\infty)$. For
$m\in(0,M)\setminus\Omega$, no solution exists.
\item[\textup{(b)}] \emph{(Smoothness.)} $\beta(m) := C(m)/m\in C^1(\Omega)$, with $\beta(m)>1$ on $\Omega$.
\item[\textup{(c)}] \emph{(Sign formula.)} On $\Omega$,
\begin{equation}\label{eq:beta-prime}
\beta'(m) \;=\; \frac{\mathcal{J}(m)}{-m^2\,\partial_\beta\Phi(\beta(m),m)},
\qquad
\mathrm{sign}\bigl(\beta'(m)\bigr) \;=\; \mathrm{sign}\bigl(\mathcal{J}(m)\bigr),
\end{equation}
with $\partial_\beta\Phi(\beta(m),m)<0$ so the denominator
of~\eqref{eq:beta-prime} is strictly positive.
\item[\textup{(d)}] \emph{(Monotonicity transfer.)} If $h$ is non-decreasing and not
constant on $(0,m)$, then $\beta'(m)>0$. If $h$ is non-increasing and not
constant on $(0,m)$, then $\beta'(m)<0$.
\item[\textup{(e)}] \emph{(Critical point under unimodality.)} Suppose
Assumption~\ref{ass:U} holds and:
\begin{itemize}
\item[\textup{(H1)}] there exists $m^\sharp\in(\eta^*_h, M)$ with $\mathcal{J}(m^\sharp)<0$;
\item[\textup{(H2)}] $[\eta^*_h, m^\sharp]\subset\Omega$.
\end{itemize}
Then there exists $m^\star\in(\eta^*_h, m^\sharp)$ with $\mathcal{J}(m^\star)=0$,
equivalently $\beta'(m^\star)=0$.
\end{enumerate}
Beyond Theorem~\ref{thm:structure} we conjecture a uniqueness statement
under log-concavity of $h$; see Conjecture~\ref{conj:uniqueness} and
Section~\ref{sec:numerics}.
\end{theorem}

\begin{remark}[Reading the sign formula]\label{rem:sign-formula-comments}
The integrand of $\mathcal{J}(m)$ in either form of~\eqref{eq:J-def}
is continuous on $[0,1]$: the denominator $\beta(m)-t\ge\beta(m)-1>0$
is bounded away from zero on $\Omega$ and $h'$ is bounded by (D1).
The prefactor $1/|\partial_\beta\Phi(\beta(m),m)|$
in~\eqref{eq:beta-prime} becomes large as $m\to\partial\Omega$ (where
$\beta(m)\downarrow 1$), so the implicit-function setup degenerates at
the boundary of the domain.
\end{remark}

The per-point monotonicity in part (d) lifts to a sub-interval statement:

\begin{corollary}[Sub-interval monotonicity]\label{cor:subinterval}
If $h$ is non-decreasing and not constant on $(0,m)$ for every $m$ in a
sub-interval $J\subset\Omega$ (in particular, if $h$ is non-decreasing and
not constant on $(0,\sup J)$), then $\beta$ is strictly increasing on $J$.
Symmetrically, if $h$ is non-increasing and not constant on $(0,m)$ for
every $m\in J$, then $\beta$ is strictly decreasing on $J$.
\end{corollary}

\begin{remark}[On (H1)]\label{rem:H1-genuine}
Identifying analytical sufficient conditions on $\rho$ that imply
(H1) -- for instance, via asymptotics of $\beta(m)$ as $m\to M$ --
is left open; see Section~\ref{sec:discussion}.
\end{remark}

\begin{proposition}[Behavior of $\beta$ at $\partial\Omega$]\label{rem:boundary}
If $m_n\in\Omega$ with $m_n\to m_\infty\in\partial\Omega\cap[0,M]$ and
$\mathcal{I}^+(m_\infty)=1$, then $\beta(m_n)\to 1$.
\end{proposition}

\begin{proof}
Passing to a subsequence, either $\beta(m_n)\to\beta_\infty\in[1,\infty)$
or $\beta(m_n)\to\infty$.

\emph{Case $\beta(m_n)\to\infty$.} Equivalently $C_n=\beta(m_n)m_n\to\infty$.
Fix $C^*\ge 2M$ large enough that $F(C^*, m_\infty)<0$ (possible by
Lemma~\ref{lem:F-boundary}(ii)). By continuity of $F$ in $m$,
$F(C^*, m_n)\to F(C^*, m_\infty)<0$. Since $C_n\to\infty$, eventually
$C_n>C^*$, and then $F(C_n,m_n)\le F(C^*,m_n)<0$ by monotonicity of
$F$ in $C$ (Lemma~\ref{lem:F-smooth}(ii)), contradicting $F(C_n,m_n)=0$.

\emph{Case $\beta_\infty>1$.} Then
$F(\beta_\infty m_\infty, m_\infty)
<\sup_{C>m_\infty}F(C,m_\infty)=\mathcal{I}^+(m_\infty)-1=0$,
again contradicting $F(C(m_n),m_n)=0\to F(\beta_\infty m_\infty, m_\infty)$.
\end{proof}

\section{Proof of Theorem~\ref{thm:structure}}\label{sec:proof}

The proof draws on Lemmas \ref{lem:I-finite}--\ref{lem:Omega-open}.
Part (a) is the implicit existence statement: for each $m\in\Omega$,
$F(\cdot,m)$ is strictly decreasing on $(m,\infty)$
(Lemma~\ref{lem:F-smooth}) and crosses zero between its
$+\infty$-side limit $\mathcal{I}^+(m)-1>0$
(Lemma~\ref{lem:F-boundary}(i)) and its limit $-1$ at infinity
(Lemma~\ref{lem:F-boundary}(ii)), so $C(m)$ exists and is unique.
Smoothness in (b) is then the implicit-function theorem. The sign
formula (c) comes from differentiating
$G(\beta(m),m)\equiv 0$ in the reparametrisation~\eqref{eq:Phi-def}
of Lemma~\ref{lem:reparam}. Parts (d) and (e) follow from (c).

\subsection{Proof of (a): domain}

Openness of $\Omega$ is Lemma~\ref{lem:Omega-open}.

For $m\in\Omega$: Lemma~\ref{lem:F-smooth}(ii) gives $F(\cdot,m)$ strictly
decreasing on $(m,\infty)$. Lemma~\ref{lem:F-boundary}(i) gives
$\lim_{C\downarrow m}F=\mathcal{I}^+(m)-1>0$. Lemma~\ref{lem:F-boundary}(ii)
gives $F\to-1<0$. By continuity and strict monotonicity, a unique
$C(m)\in(m,\infty)$ satisfies $F(C(m),m)=0$.

For $m\in(0,M)\setminus\Omega$: $\sup_{C>m}F(C,m)=\mathcal{I}^+(m)-1\le 0$.
Since $F(\cdot,m)$ is strictly decreasing in $C$ (Lemma~\ref{lem:F-smooth}(ii)),
the supremum equals the limit as $C\downarrow m$ and is not attained. Hence
$F(C,m)<\mathcal{I}^+(m)-1\le 0$ for all $C>m$, and the equation $F(C,m)=0$
has no solution.

\subsection{Proof of (b): smoothness and bound}

Fix $m_0\in\Omega$. Since $C(m_0)>m_0$ by part (a), the point $(C(m_0), m_0)$
lies in the open domain $\{(C,m):C>m,\,0<m<M\}$ where $F\in C^1$
(Lemma~\ref{lem:F-smooth}). Since $\Omega$ is open (part (a)), there is an open
neighborhood $V\subset\Omega$ of $m_0$. By Lemma~\ref{lem:F-smooth},
$F\in C^1$ near $(C(m_0),m_0)$, and $\partial_C F<0$ there. The implicit
function theorem \citep[see e.g.][]{krantz-parks-2002} produces a $C^1$
function $\widetilde C$ on some open
$V_0\subset V$ containing $m_0$, with $F(\widetilde C(m),m)=0$ and
$\widetilde C(m_0)=C(m_0)$. Uniqueness of the solution at each $m\in V_0$
(part (a)) forces $\widetilde C=C$ on $V_0$. Since $m_0$ is arbitrary,
$C\in C^1(\Omega)$, hence $\beta=C/m\in C^1(\Omega)$.

Bound $\beta(m)>1$: $C(m)>m$ by part (a).

\subsection{Proof of (c): sign formula}

We work with the reparametrized form (Lemma~\ref{lem:reparam}): on $\Omega$,
$G(\beta(m),m)=0$ where $G(\beta,m):=m\Phi(\beta,m)-1$.

\emph{Step 1: Setup.} The proof of (c) uses the $C^1$ regularity of $\beta$ established in (b); the chain rule below applies. Fix $m_0\in\Omega$. By openness ($\Omega$ open) and part
(b) ($\beta$ continuous), there is a compact neighborhood $K\subset\Omega$ of
$m_0$ on which $\beta_K:=\min_{m\in K}\beta(m)>1$ is attained. Pick
$\beta_-\in(1,\beta_K)$ and $\beta_+>\max_{m\in K}\beta(m)$.

\emph{Step 2: Partial derivatives of $G$.}
On the rectangle $[\beta_-,\beta_+]\times K$ the denominator
$\beta-t$ is bounded below by $\beta_- - 1>0$, and $\rho$ and $h'$ are
bounded by (D1). The pointwise $\beta$- and $m$-derivatives of the
integrand of $G+1$ are therefore dominated by constants
($\|\rho\|_\infty/(\beta_- - 1)^2$ and $\|h'\|_\infty/(\beta_- - 1)$
respectively), so differentiation under the integral sign applies
and gives
\begin{align}
\partial_\beta\Phi(\beta,m) &= -\int_0^1\frac{t\,\rho(mt)}{(\beta-t)^2}\,dt \;<\;0,\label{eq:dbeta-Phi}\\
\partial_m G(\beta,m) &= \phantom{-}\int_0^1\frac{t\,h'(mt)}{\beta-t}\,dt,\label{eq:dm-G}
\end{align}
both continuous in $(\beta,m)$ by dominated convergence. The strict
inequality $\partial_\beta\Phi<0$ uses (D3). Hence $G\in C^1$ on
$[\beta_-,\beta_+]\times K$, and $\partial_\beta G=m\,\partial_\beta\Phi<0$.

\emph{Step 3: Identify $\partial_m G$ with $\mathcal{J}(m)/m$.}
The second form of~\eqref{eq:J-def} gives directly
\begin{equation}\label{eq:dm-G-J}
\partial_m G(\beta(m),m) \;=\; \int_0^1 \frac{t\,h'(mt)}{\beta(m)-t}\,dt
\;=\; \frac{\mathcal{J}(m)}{m}.
\end{equation}

\emph{Step 4: Chain rule.} Differentiating $G(\beta(m),m)\equiv 0$
at $(\beta(m),m)$ gives
$0 = \partial_\beta G\cdot\beta'(m) + \partial_m G$, so
\begin{equation}\label{eq:chain-rule}
\beta'(m) \;=\; -\frac{\mathcal{J}(m)/m}{m\,\partial_\beta\Phi(\beta(m),m)}
\;=\; \frac{\mathcal{J}(m)}{-m^2\,\partial_\beta\Phi(\beta(m),m)}.
\end{equation}
The prefactor is strictly positive on $\Omega$ since $\partial_\beta\Phi<0$,
so $\mathrm{sign}(\beta'(m))=\mathrm{sign}(\mathcal{J}(m))$.

\subsection{Proof of (d): monotone case}

Suppose $h$ is non-decreasing and not constant on $(0,m)$. By (D1), $h\in C^1$,
so $h'\ge 0$ everywhere on $(0,m)$ (a $C^1$ non-decreasing function has
non-negative derivative pointwise).

If $h'\equiv 0$ on $(0,m)$, then $h$ is constant, contradicting the
non-constancy assumption. Since $h$ is non-constant on $(0,m)$, there exist
$a<b$ in $(0,m)$ with $h(a)<h(b)$; by the mean value theorem, there is
$\eta_1\in(a,b)$ with $h'(\eta_1)=(h(b)-h(a))/(b-a)>0$. By continuity of
$h'$ (which is in $C^0$ since $h\in C^1$), there is an open interval
$I\subset(0,m)$ containing $\eta_1$ on which $h'\ge\delta>0$ for some
$\delta$. The set $T:=\{t\in(0,1):mt\in I\}$ is an open
subinterval of $(0,1)$ of positive Lebesgue measure.

In $\mathcal{J}(m)=m\int_0^1 t\,h'(mt)/(\beta(m)-t)dt$ (second form of~\eqref{eq:J-def}):
\begin{itemize}
\item the prefactor $m>0$;
\item the weight $t/(\beta(m)-t)>0$ for $t\in(0,1]$ (since $\beta(m)>1\ge t$);
\item $h'(mt)\ge 0$ everywhere;
\item $h'(mt)\ge\delta$ on $T$.
\end{itemize}
Hence the integrand is $\ge 0$ a.e.\ and strictly positive on the
positive-measure set $T$, giving $\mathcal{J}(m)>0$. By part (c),
$\beta'(m)>0$. The non-increasing case is symmetric.

\subsection{Proof of (e): conditional critical point}

\emph{Step 1: $\mathcal{J}(\eta^*_h)>0$.} By (H2), the closed interval
$[\eta^*_h,m^\sharp]\subset\Omega$; in particular $\eta^*_h\in\Omega$, so
by part (b) $\beta(\eta^*_h)>1$ is well-defined and $\mathcal{J}(\eta^*_h)$
is well-defined. By Assumption~\ref{ass:U}, $h$ is strictly increasing on
$(0,\eta^*_h)$, hence non-decreasing and not constant on $(0,\eta^*_h)$.
The hypothesis of part (d) (with $m=\eta^*_h$) is thus satisfied, and
applying (d) at $m=\eta^*_h$ yields $\beta'(\eta^*_h)>0$, equivalently (by
(c)) $\mathcal{J}(\eta^*_h)>0$.

\emph{Step 2: Sign change.} (H1) gives $\mathcal{J}(m^\sharp)<0$.

\emph{Step 3: Continuity of $\mathcal{J}$ on $[\eta^*_h,m^\sharp]$.}
By (H2) the compact interval $[\eta^*_h,m^\sharp]$ lies in $\Omega$,
and by~(b) $\beta$ is continuous there, so
$\beta_-:=\min_{[\eta^*_h,m^\sharp]}\beta(m)>1$. On this interval
the integrand of $\mathcal{J}(m)=m\int_0^1 t\,h'(mt)/(\beta(m)-t)\,dt$
is uniformly bounded by $m^\sharp\,\|h'\|_\infty/(\beta_-{-}1)$.
Dominated convergence and continuity of the prefactor $m$ then give
$\mathcal{J}\in C^0([\eta^*_h,m^\sharp])$.

\emph{Step 4: Intermediate value theorem.} Since $\mathcal{J}$ is
continuous on $[\eta^*_h,m^\sharp]$ with $\mathcal{J}(\eta^*_h)>0$
and $\mathcal{J}(m^\sharp)<0$, there exists
$m^\star\in(\eta^*_h,m^\sharp)$ with $\mathcal{J}(m^\star)=0$, hence
$\beta'(m^\star)=0$ by~(c).

\subsection{A uniqueness conjecture and a bimodal counterexample}

Theorem~\ref{thm:structure} gives no uniqueness statement for the
critical point produced in part~(e). We conjecture uniqueness under
strict log-concavity of $h$, and we examine numerically a bimodal
density showing what can happen when this hypothesis fails.

\begin{conjecture}[Uniqueness under log-concavity]\label{conj:uniqueness}
Assume \textup{(D0)--(D3)}, hypotheses \textup{(H1), (H2)} of
Theorem~\ref{thm:structure}(e), and that $h(\eta):=\eta\rho(\eta)$ is
positive on $(0,M)$ and strictly log-concave:
\begin{equation}\label{eq:logconcave}
\log h\bigl(\lambda x + (1-\lambda) y\bigr) \;>\; \lambda \log h(x) + (1-\lambda)\log h(y)
\qquad (x\ne y\text{ in }(0,M),\ \lambda\in(0,1)).
\end{equation}
Then the critical point $m^\star$ produced by
Theorem~\ref{thm:structure}(e) is the unique critical point of
$\beta$ on $\Omega$.
\end{conjecture}

\begin{remark}[On the hypotheses of Conjecture~\ref{conj:uniqueness}]
\label{rem:conj-hypotheses}
The positivity $h>0$ on $(0,M)$ is stated explicitly because, without
it, $h(\eta)=0$ would make $\log h=-\infty$ and~\eqref{eq:logconcave}
trivially true. Strict log-concavity together with $h(0)=h(M)=0$
already implies unimodality, which is therefore not listed
separately. Hypothesis (H1) is not implied by log-concavity; we
verify it numerically in all seven log-concave test densities of
Table~\ref{tab:critical}. Finally, $h>0$ on $(0,M)$ implies $\rho>0$
there, so Lemma~\ref{lem:I-finite}(ii) gives $\Omega=(0,M)$ and the
conjecture is a statement on the full interval.
\end{remark}

The seven log-concave rows of Table~\ref{tab:critical} support the
conjecture: in each case $\beta'$ has a single sign change on the
tested portion of $\Omega$. The examples span a range of vanishing
orders ($\alpha\in\{3,4,5\}$), polynomial and non-polynomial
densities, and modes throughout $(0,M)$. The list is biased toward
Beta-type densities; a systematic test of non-Beta-type log-concave
densities (for instance, polynomial-times-trigonometric or
polynomial-times-logarithmic) would strengthen the numerical case.

To probe what fails when log-concavity is dropped, we examine the
following bimodal density.

\begin{example}[Bimodal $\rho$ violating both unimodality and log-concavity]\label{rem:bimodal}
For $\sigma\in(0, 0.1]$ (we take $\sigma=0.05$) let
\begin{equation}\label{eq:bimodal-rho}
\rho(\eta) \;\propto\; \eta\,(1-\eta)^2\,\bigl[e^{-(\eta-0.3)^2/(2\sigma^2)} + e^{-(\eta-0.7)^2/(2\sigma^2)}\bigr], \qquad \eta\in[0,1],
\end{equation}
normalised on $[0,1]$. Each factor in~\eqref{eq:bimodal-rho} is
smooth, so $\rho\in C^\infty([0,1])$ and (D1) holds. The factor
$\eta$ gives $\rho(0)=0$; the factor $(1-\eta)^2$ gives
$\rho(\eta)=O((1-\eta)^2)$ near $\eta=1$, so (D2) holds with
$\alpha=3$; (D0) and (D3) are immediate. Setting
$h:=\eta\rho(\eta)$,
\begin{equation}\label{eq:bimodal-h}
h(\eta) \;\propto\; \eta^{2}(1-\eta)^{2}\,
\bigl[e^{-(\eta-0.3)^2/(2\sigma^2)}+e^{-(\eta-0.7)^2/(2\sigma^2)}\bigr],
\end{equation}
and both factors are invariant under $\eta\mapsto 1-\eta$, so
$h(\eta)=h(1-\eta)$; differentiating this symmetry at $\eta=1/2$
gives $h'(1/2)=0$.

We claim $h$ is bimodal at $\sigma=0.05$. Since $h\in C^\infty([0,1])$
with $h(0)=h(1)=0$ and $h>0$ on $(0,1)$, $h$ attains its maximum at
some interior point $\eta_M$. If $\eta_M=1/2$, the symmetry would
make $1/2$ the unique global maximum and $h$ unimodal. This is ruled
out numerically: at 50-digit precision $h(1/2)\approx 1.61\times 10^{-3}$
while $h(0.30895)\approx 1.7249$. By symmetry $1-\eta_M$ is then a
second maximum, and the critical point at $1/2$ is the local minimum
between them. A direct search at the same precision finds exactly
three critical points of $h$ on $(0,1)$,
\begin{equation}\label{eq:bimodal-crit}
\eta \;\in\;\{0.30895,\;0.5,\;0.69105\}.
\end{equation}
The same picture holds at $\sigma=0.06$ (three sign changes of
$\mathcal{J}$ at $m^\star\approx 0.370, 0.552, 0.751$), while at
$\sigma=0.04$ the middle root of $\mathcal{J}$ falls below our
numerical tolerance and we report three sign changes only for
$\sigma\ge 0.05$.

Since $\rho>0$ on $(0,1)$, Lemma~\ref{lem:I-finite}(ii) gives
$\Omega=(0,1)$ and Theorem~\ref{thm:structure}(a)--(d) applies;
part~(e) and Conjecture~\ref{conj:uniqueness} do not, because $h$ is
neither unimodal nor log-concave.
\end{example}

Numerically (row~8 of Table~\ref{tab:critical}), $\beta$ has three
critical points in this example, in contrast with the single
critical point observed under log-concavity. The example shows that
dropping unimodality and log-concavity jointly admits multiple
critical points, but it does not separate the roles of the two
hypotheses: a unimodal but non-log-concave $h$ producing multiple
critical points of $\beta$ would settle which of the two is
operative, and we leave that question open
(Section~\ref{sec:discussion}).

\section{Numerical evidence}\label{sec:numerics}

\subsection{Method}

For each density we solve $F(C,m)=0$ on a uniform grid of $N=2001$
points $m_i\in[0.1, M-10^{-3}]$ using Brent's method
(\texttt{scipy.optimize.brentq}, tolerance $10^{-12}$). The integral
in $F$ is computed by adaptive quadrature
(\texttt{scipy.integrate.quad}, relative tolerance $10^{-10}$).
We then evaluate $\mathcal{J}(m)$ by the same quadrature using the
analytical formula~\eqref{eq:beta-prime} of
Theorem~\ref{thm:structure}(c), which avoids numerical
differentiation. Sign changes of $\mathcal{J}$ between consecutive
grid points are flagged and refined by Brent's method on
$\mathcal{J}$ to tolerance $10^{-12}$. For row~8 of
Table~\ref{tab:critical} we recompute all reported critical points
with \texttt{mpmath} at 50-digit precision; this reproduces the
double-precision values to 8 digits.

We verified robustness of the sign-change count: doubling the grid
to $N=4001$ leaves the count unchanged and shifts each $m^\star$ by
$<10^{-5}$; replacing adaptive quadrature with 64-point
Gauss--Legendre shifts $m^\star$ by $<10^{-7}$; varying the
sign-change threshold over $[10^{-10},10^{-6}]$ leaves the count
unchanged. The lower-endpoint $m_{\min}=0.1$ is well inside $\Omega$
for all seven log-concave densities ($\mathcal{I}^+(m_{\min})>10$),
hence away from the boundary regime where the prefactor
in~\eqref{eq:beta-prime} degenerates (Proposition~\ref{rem:boundary}).

\subsection{Results}

\begin{table}[h]
\centering
\caption{Critical-point structure of $\beta(m)$. In rows~1--7
($\log$-concave $h$ satisfying (D0)--(D3)) the derivative $\beta'$
has exactly one sign change on the tested portion of $\Omega$, so
each row records a single $m^\star$ and the displacement
$\textrm{Gap}:=m^\star-\eta^*_h$ of the critical point from the mode
of $h$. Row~8 is the bimodal density of Example~\ref{rem:bimodal},
for which $\beta'$ has three interior sign changes; the
``Gap'' column is left blank there. Densities are written up to
normalisation; for $\eta^a(1-\eta)^b$ on $[0,1]$ the normalising
constant is $1/B(a+1,b+1)$.}
\label{tab:critical}
\begin{tabular}{lcccc}
\toprule
Density on $[0,1]$ & $\alpha$ & $\eta^*_h$ & $m^\star$ & Gap \\
\midrule
\multicolumn{5}{l}{\emph{Log-concave $h$:}} \\
$\eta^0(1-\eta)^2$ \emph{(Beta)} & 3 & 0.3333 & 0.4289 & 0.0955 \\
$\eta^1(1-\eta)^2$ & 3 & 0.5000 & 0.6381 & 0.1381 \\
$\eta^1(1-\eta)^3$ & 4 & 0.4000 & 0.5181 & 0.1181 \\
$\eta^2(1-\eta)^3$ & 4 & 0.5000 & 0.6354 & 0.1354 \\
$\eta^4(1-\eta)^4$ & 5 & 0.5556 & 0.6900 & 0.1345 \\
$\eta(1-\eta)^2 e^{-2\eta}$ & 3 & 0.3820 & 0.5009 & 0.1190 \\
$\cos^2(\pi\eta/2)$ & 3 & 0.4159 & 0.5267 & 0.1108 \\
\midrule
\multicolumn{5}{l}{\emph{Bimodal $h$ (Example~\ref{rem:bimodal}, $\sigma=0.05$):}} \\
\multicolumn{5}{l}{$\rho(\eta)\propto\eta(1-\eta)^2[e^{-(\eta-0.3)^2/(2\sigma^2)}+e^{-(\eta-0.7)^2/(2\sigma^2)}]$} \\
\quad bimodal $h$ ($\alpha=3$) & 3 & $0.3090,\,0.6911$ & $0.3620,\,0.5455,\,0.7502$ & --- \\
\bottomrule
\end{tabular}
\end{table}

\paragraph{Row~8: the bimodal example.}
The bimodality argument of Example~\ref{rem:bimodal} gives two local
maxima of $h$ at $\eta=0.30895, 0.69105$ (with $h=1.72488$, equal by
reflection symmetry) separated by a local minimum at $\eta=0.5$ with
$h=0.00161$ (peak-to-trough ratio $\approx 1075$). Running the
protocol of Section~\ref{sec:numerics} yields three interior sign
changes of $\mathcal{J}$ on $\Omega=(0,1)$ at
\begin{equation}\label{eq:bimodal-mstar}
m^\star \in \{0.36196,\;0.54547,\;0.75023\},
\quad
\beta(m^\star) \in \{1.34002,\;1.00051,\;1.28195\}.
\end{equation}
The middle root sits at $\beta\approx 1.00051$, close to the
boundary $\beta=1$ where the prefactor in~\eqref{eq:beta-prime}
becomes large; it would be reasonable to suspect a numerical
artefact there.

\paragraph{Verification of the middle root.}
A fine-grid diagnostic gives
$\mathcal{J}(0.540)\approx -0.894$,
$\mathcal{J}(0.545)\approx -0.074$,
$\mathcal{J}(0.550)\approx +0.689$, stable across quadrature
tolerances $10^{-10}$--$10^{-14}$. An independent \texttt{mpmath}
computation at 50-digit precision reproduces all three roots
in~\eqref{eq:bimodal-mstar} to 8 digits. The middle root is real.

Three sign changes are qualitatively consistent with $h$ being
bimodal but not formally implied by it; deriving the count of zeros
of $\mathcal{J}$ from the local-extremum structure of $h$ is open
(Section~\ref{sec:discussion}). Parts (a)--(d) of
Theorem~\ref{thm:structure} apply to row~8 directly; only part (e)
and Conjecture~\ref{conj:uniqueness} are inapplicable.

\section{Discussion}\label{sec:discussion}

\subsection{Bianconi--Barab\'asi networks and condensation}

In the Bianconi--Barab\'asi model \citep{bianconi-barabasi-2001},
$\rho$ is the fitness density of nodes, $m$ is the maximum fitness,
and $C(m)$ is the normalisation in the rate equation governing the
growth of node degrees; the ratio $\beta=C/m$ controls the resulting
degree distribution. The threshold $\mathcal{I}^+(m)=1$ is
conjecturally related to the Bose--Einstein condensation transition
in the deterministic mean-field reduction. The relationship to the
stochastic model is delicate: there, $m$ is a random order
statistic, and most fitness densities of empirical interest
(log-normal, exponential) fail (D2). For rigorous probabilistic
treatments of preferential-attachment models with fitness, see
\citet{borgs-chayes-daskalakis-roch-2007} and
\citet{dereich-moerters-2009}. Theorem~\ref{thm:structure}(a)
characterises when the deterministic mean-field equation admits a
solution; the sign formula (c) describes how the structural ratio
depends qualitatively on the upper fitness limit. A rigorous
reduction of the stochastic model to~\eqref{eq:F-def} is beyond the
scope of this paper.

\subsection{Reinforcement-driven selection}

Selection rules of the form $P(i)\propto x_i^\beta$, with $x_i$ an
accumulated score, arise in several settings, including the
nonlinear preferential-attachment networks of
\citet{krapivsky-redner-2001}. When a self-consistency equation of
the form~\eqref{eq:F-def} emerges as the mean-field reduction in
such a setting, the qualitative behaviour of $\beta$ is described by
Theorem~\ref{thm:structure}.

\subsection{The (D2) restriction}\label{sec:discussion-D2}
The principal limitation of the present analysis is hypothesis (D2),
which excludes the densities most relevant in network applications
(log-normal, exponential, Gamma). A natural approach is to use the
boundary regularisations
$\rho^{(R,\alpha)}(\eta)=\rho(\eta)(R-\eta)^{\alpha-1}/Z_{R,\alpha}$,
which satisfy (D2) for finite $R$, and to study the limit
$R\to\infty$. Whether the limit
$m^\star_\alpha:=\lim_{R\to\infty}m^\star_{R,\alpha}$ exists for
log-normal $\rho$ is open.

\section*{Acknowledgements}

The author thanks Kirill Borodin (Moscow Technical University of
Communications and Informatics, Moscow, Russian Federation;
\texttt{k.n.borodin@mtuci.ru}) for proof-reading the manuscript and
for several suggestions on presentation.

\end{document}